\documentclass[10pt]{article}

\usepackage{amsmath,amssymb,amstext} 
\usepackage{caption}
\makeatletter
\newif\if@restonecol
\makeatother

\usepackage{algorithm2e}
\usepackage{graphicx}
\usepackage{url}
\usepackage{amsthm}
\usepackage{verbatim}
\usepackage{subfigure}
\usepackage{soul}
\usepackage{color}
\usepackage{amsfonts,latexsym,amstext}
\usepackage{setspace}
\usepackage{eepic}
\newtheorem{definition}{Definition}

\newtheorem{corollary}{Corollary}

\newtheorem{theorem}{Theorem}
\newtheorem{example}{Example}

\newcommand{\A}{\mathcal{A}}

\title{Matching Games with Additive Externalities}

\author{
Simina Br\^{a}nzei
\footnote{Aarhus University}
\and
Tomasz P. Michalak
\footnote{Oxford University}
\and
Talal Rahwan
\footnote{University of Southampton}
\and
Kate Larson
\footnote{University of Waterloo}
\and
Nicholas R. Jennings
\footnote{University of Southampton}
}

\begin{document}

\maketitle

\begin{abstract}
Two-sided matchings are an important theoretical tool used to model markets and social interactions.
In many real life problems the utility of an agent is influenced not only by their own choices, 
but also by the choices that other agents make. Such an influence is called an externality. 
Whereas fully expressive representations of externalities in matchings require exponential space,
in this paper we propose a compact model of externalities, in which the influence of a 
match on each agent is computed additively. In this framework, we analyze many-to-many and
one-to-one matchings under neutral, optimistic, and pessimistic behaviour, and provide both 
computational hardness results and polynomial-time algorithms for computing stable outcomes.
\end{abstract}

\vspace{-2mm}
\section{Introduction}\label{sec:Introduction}
Matching games are an important theoretical abstraction which have been extensively studied in several fields, including economics, combinatorial
optimization, and computer science. Matchings are often used to model markets, and examples include the classical marriage problem, firms and
workers, schools and students, hospitals and medical interns~\cite{Burkard09:AssignmentProblems}.

Previous matching literature focused primarily on one-to-one and one-to-many models ~\cite{RothSotomayor90:TwoSidedMatchings}.
More recently, however, attention has been paid to more complex models of many-to-many matchings due to their relevance to real-world situations. For example, most labour markets involve at least a few many-to-many contracts~\cite{Echenique06:TheoryStability}. More realistic matching models should take into account the fact that in many settings the utility of an agent is influenced not only by their own choices, but also by the choices that other agents make.
Such an influence is called an externality. For instance, companies care not only about the employers they hire themselves, but also about the employers hired by other companies. This aspect is crucial to how competitive a company is on the market, and so externalities must be considered to completely understand such situations. Modelling matchings with externalities is computationally challenging, as fully expressive representations require exponential space.

One of the central questions in matching games is stability.
There exists a rich literature on the stability of matchings (see, e.g., \cite{RothSotomayor90:TwoSidedMatchings}) which dates back to the classical stable marriage problem~\cite{GaleShapley62:StabilityOfMarriage}. In the presence of externalities, stability becomes much more complex and challenging phenomenon due to the fact that a deviation of some agents typically affects utilities of the residual agents --- thus, invoking a response from them that can change the worth of the deviation dramatically. This key issue has been heavily studied in the context of coalitional games.\footnote{\footnotesize In many respects, coalitional games are closely related to matching games.} Indeed, already in the seminal work of Neumann and Morgenstern (1944) the characteristic function form of a coalitional game was defined using the minimax rule. This definition assumes 
absolute pessimism of deviators who expect that residual agents hurt them to the maximum possible extent. Since then this idea has been followed by many other authors in the literature on coalitional games (see \cite{Koczy09:GEB}). This includes also Thrall and Lucas introductory paper on coalitional games with externalities \cite{Lucas:Thrall63}. On the other hand, Shapley and Shubik, among others, considered a complementary \emph{optimistic} assumption in which residual agents maximise the deviators' utility \cite{ShapleyShubik1966}. Between these two extremes a \emph{neutral} approach assumes no reaction from residual agents \cite{Koczy07:Recursive}.

Our contribution in this paper can be summarized as follows.
We formulate a compact model of externalities for matchings,
in which the influences of matches on the agents are computed additively.
Additive functions are a natural tradeoff between expressivity and computational attractiveness, and have been used in other domains,
such as additively separable hedonic games~\cite{Jackson02}.
Second, we consider key stability concepts, namely setwise stability in the case of many-to-many matchings
and setwise and pairwise stability in the case of one-to-one matchings, under optimistic, neutral and pessimistic reasoning, inspired by the game theoretic literature.
We study the computational properties of these stable sets, provide both hardness results and polynomial algorithms where applicable, and show how the sets are related to each other.

\vspace{-2mm}
\section{The Model}\label{sec:Model}

Let $N = M \cup W$ be the set of agents, where $M  = \{m_1, \ldots, m_{|M|}\}$ and $W = \{w_1, \ldots, w_{|W|}\}$ are disjoint.
A \emph{match}, $(m,w)$ is an edge between two agents $m\in M$ and $w\in W$.  We let a matching, $\A$, be a set
of all matches. 
If the number of allowable matches any agent can participate in is unrestricted then we say we have a \emph{many-to-many matching problem},
while if each agent can participate in at most one match then we have a \emph{one-to-one matching problem}.
We assume the formation of a match requires the consent of both parties, while severing a match can be done unilaterally by any of its
endpoints. The \emph{empty matching} contains no matches, while the \emph{complete matching} contains all the possible matches.
A matching game with externalities is defined as follows:

\begin{definition}[Matching Game with Externalities] \label{def:matching_game_externalities}
A \emph{matching game with externalities} is represented as a tuple $G = (M, W, \Pi)$, where $(M,W)$ is the set of agents 
and $\Pi$ is a real valued function such that $\Pi(\A|z)$ is the utility of agent $z$ when matching $\A$ forms.
\end{definition}

The model is related to the assignment game of Shapley and Shubik \cite{ShapleyShubik1971}, in which $\Pi(m,w)$ represents the monetary gains when $m$ and $w$ are matched, amount which should be distributed among the agents. It is also related to the generalization of the assignment game to include externalities by Sasaki and Toda \cite{SasakiToda96:MarriageExternalities}, in which for each matching $A$ and every $(m,w) \in A$, the amount $\Pi(m,w|A)$ represents the gains from matching of the pair $(m,w)$ at $A$.

The game belongs to the class of non-transferable utility (NTU) games in partition function form, where the valuation of each agent depends on the underlying coalition structure. Thus Definition \ref{def:matching_game_externalities} gives the most general form of a matching game with externalities when utility is non-transferable.
%
%
 For a more detailed treatment of games in partition function form, we refer the reader to Chalkiadakis \emph{et al} \cite{coopbook}, Myerson \cite{Myerson:values_pfg}, and Finus \cite{Finus01}. In addition, there exists an entire literature about externalities in different coalition formation contexts (see, e.g. Bloch \cite{Bloch96}, Ray and Vohra \cite{RayVohra}, Maskin \cite{Maskin}, Michalak \emph{et al} \cite{Michalak}).

We are interested in settings where an agent's utility is formed by additive externalities.

\begin{definition}[Matching Game with Additive Externalities]
A \emph{matching game with additive externalities} is represented as a tuple $G = (M,W,\Pi)$, where
$(M,W)$ is the set of agents and $\Pi$ is a real valued function such that $\Pi(m,w|z)$ is the value that 
agent $z$ receives from the formation of match $(m,w)$.
Given a matching $\A$ over $N$, the \emph{utility} of an agent $z$ in $\A$ is:
$u(z, \A) = \sum_{(m,w) \in \A} \Pi(m, w | z)$.
\end{definition}

Thus, in a matching game with additive externalities, an agent's utility is the sum of values it receives from matches it 
participates in, along with the sum of all externalities that arise due to the matchings of other agents.
The additive representation for matching games with externalities is related to an important succinct representation of hedonic 
games (see Bogomolnaia and Jackson \cite{Bogomolnaia}, Aziz, Brandt, and Seedig~\cite{Aziz11:OptimalHedonic,Aziz11:StableHedonic}).
An additively separable hedonic game is represented by a weighted graph such that every node is an agent, and the
value of an edge $a(i,j)$ represents the value that agent $i$ receives from $j$ when the two agents are in the same
coalition. There has been a recent surge of literature on the complexity of computing stable and optimal outcomes in additively separable hedonic games.
However, we are not aware of any study regarding computational aspects of externalities in either matching games or additively separable hedonic games.

We are interested in whether matchings are \emph{stable}, and whether there exist stable matchings given a particular matching game $G$.
In general, a matching is stable if no subset of agents has any incentive to reorganize and form new matchings amongst themselves.
We distinguish between three standard stability concepts which commonly appear in the matching literature. The first, 
setwise stability, is the most general and encompasses the other two (corewise stability and pairwise stability).
Unless otherwise noted, the stability concept used in this paper is setwise stability, which we interchangeably refer to as set stability.

\begin{definition}[Setwise Stable]
Given a matching game $G=(M,W,\Pi)$, a matching $\A$ of $G$ is \emph{setwise stable} if there does not exist a set of agents $B\subseteq N$,
which can improve the utility of at least one member of $B$ while not degrading the others by:
\begin{itemize}
\item rearranging the matches among themselves
\item deleting a (possibly empty) subset of the matches with agents in $N \setminus B$.
\end{itemize}
\end{definition}
If such a coalition $B$, exists, then it is called a \emph{blocking coalition}.

\begin{definition}[Corewise Stable]
Given a matching game $G=(M,W,\Pi)$, a matching $\A$ of $G$ is \emph{corewise stable} if there does not exist a set of agents $B\subseteq N$,
which can improve the utility of at least one member of $B$ while not degrading the others by:
\begin{itemize}
\item rearranging the matches among themselves 
\item deleting all the matches with agents in $N \setminus B$.
\end{itemize}
\end{definition}

In matching games without externalities, the actions taken by other agents have only a limited effect on an agent -- its utility depends solely on 
who the agent is matched with, and does not depend on matches involving others.
However, if there are externalities then this is no longer true. 
The utility of agent $m$ for example, may depend on the matches involving agent $w$ even if $(m,w)\not\in\A$.  Therefore, we argue, the stability concepts 
need to account for the actions taken by agents in $N\setminus B$ after a deviation by coalition $B$, since these actions will effect the final utility of 
the deviating agents in $B$.
However, it may be hard to compute the possible reactions to a deviation since there are potentially an exponential number (i.e. all possible matchings 
amongst agents in $N \setminus B$).  
Instead, in this paper we consider several natural heuristics that deviating agents in $B$ can use to reason about, and approximate, the reactions to 
their possible deviations.
\begin{itemize}
\item Agents in blocking coalition $B$ may take a \emph{neutral} perspective. That is, they assume that agents in 
$N \setminus B$ will not react to the deviation.  All existing matches amongst non-deviating agents will remain and no new matches will form.

\item Agents in $B$ may take an \emph{optimistic} attitude and assume that any response to their deviation will be optimal. 
Such an idea has 
been studied in the literature in coalitional games~\cite{ShapleyShubik1966,Koczy07:Recursive}.
In our model, optimism translates to the following. When coalition $B$ considers deviation $\A^{'}$ from $\A$, every agent $i \in B$ evaluates the deviation assuming the agents in $N \setminus B$ will organize themselves in the best possible way for $i$. That is, $i$ hopes that $N \setminus B$ will cut any matches with a negative influence on $i$ and form all the matches with a positive influence on $i$, including initiating 
those with $i$ as an endpoint. According to $i$'s evaluation, this leads to matching $\A^{''}$. 
If $i$ improves (or does not degrade) in $\A^{''}$ compared to $\A$, then $i$ agrees to participate in blocking.

\item Agents in $B$ may take a \emph{pessimistic} attitude and assume that the agents in $N \setminus B$ will take actions so as to punish the members of $B$. 
That is, the non-deviators will cut matches with a positive influence on the deviators, and will form any new matches with a negative influence on them.
\end{itemize}
The idea of pessimism has also been studied in the coalitional game theory literature~\cite{Lucas:Thrall63,Koczy09:GEB,Koczy07:Recursive}.
Note that with optimistic and pessimistic attitudes, there may be inconsistencies among the reactions that the members of $B$ expect from $N \setminus B$,
and thus the best or worst possible outcome may not happen for each member of $B$.




It should be noted that the optimistic and pessimistic cases are in fact the best, and worst, possible reactions to the deviating agents, respectively. As such, these cases represent \textit{upper and lower bounds} on the impact a reaction would have on the deviating agents. Based on this, we argue that the analysis of those two particular cases is essential when assessing the expected reward/risk associated with the deviation. Such bounds can also be very useful when designing optimization algorithms for matching games with externalities. In fact, the idea of computing bounds on externalities has already proven very useful for general partition function games with positive or negative externalities; this idea provided the corner stone for the currently-state-of-the-art algorithm for coalition structure generation for those games. 

It is also important to note that each member of a blocking coalition $B$ is required to perform at least one action,  
by severing a match with another agent in $N$, or by forming a new match with another agent in $B$.
More realistic definitions of stability can incorporate different layers of deviators, such as agents who 
perform the deviation and agents who agree to it without actively participating.
However, such definitions can be problematic, and require specifying which agents are identified as deviators and how they should be treated depending on their role.
For this reason, in this paper we only consider one type of deviators, those who are required to perform at 
least one action.


\vspace{-4mm}
\section{Many-to-Many Matchings}
In this section we analyze the stability of many-to-many matchings using the stable set as a solution concept.
We provide hardness for computing stable outcomes, an FPT algorithm for the optimistic stable set, and describe the relationship between the stable sets.

\vspace{-2mm}
\subsection{Neutral Stability}

The neutral set can be empty, as the following example illustrates.
\begin{example}\label{eg:empty_setwise_stable}
Let $M = \{m\}$, $W$ $=$ $\{w_1$, $w_2\}$, and $\Pi$ as follows:
$\Pi(m, w_i | m)$ $=$ $0$,
$\Pi(m, w_i|w_i)$ $=$ $\varepsilon$, and
 $\Pi(m, w_1 | w_2)$ $=$ $\Pi(m, w_2| w_1)$ $=$ $-\Delta$, where $\Delta > \varepsilon > 0$.
The empty matching is blocked by $(m, w_1)$, $\{(m, w_1)$, $(w_2)\}$ is blocked by $(m, w_2)$, $\{(m, w_2), (w_1)\}$ is blocked by $(m, w_1)$, and
$\{(m, w_1), (m, w_2)\}$ is blocked by the empty matching, thus the neutral stable set is empty.
\end{example}

The next game contrasts the case where the empty matching belongs to the neutral stable set with the case where it does not.
\begin{example}\label{eg:paradoxical_matching}
Let $G = (M, W, \Pi)$ be such that $\Pi(m,w|m)=\Pi(m,w|w)=-\varepsilon < 0$ for all $m, w$, and $\Pi(m,w|z)= \delta > 0$ for all
$z \neq m,w$. If $\varepsilon \gg \delta$, the only matching satisfying neutral set stability is the empty matching, since the formation of any match
is very expensive for the participating agents and not compensated by the utility obtained from the externalities. Otherwise, if
$\delta \gg \varepsilon$, the grand coalition can block the empty matching through the complete matching,
and so the neutral stable set is empty.
\end{example}

For neutral stability we have the following hardness results.
\begin{theorem}\label{thm:nonempty_neutral_many}
Checking nonemptiness of the neutral stable set is NP-hard.
\end{theorem}
\begin{proof}
We provide a reduction from the \textit{Knapsack} problem.
Let $I = \langle U, s, v, B, K \rangle$, 
where $U = \{u_1, \ldots, u_n\}$ is a finite set, $s(u) \in \mathbb{Z}^{+}$
the size of element $u \in U$, $v(u) \in \mathbb{Z}^{+}$ the value of element $u \in U$, $B \in \mathbb{Z}^{+}$ a size constraint,
and $K \in \mathbb{Z}^{+}$ a value goal, such that $U^{'} \subseteq U$ is a solution if
$\sum_{u \in U^{'}} s(u) \leq B \; \mbox{and}\; \sum_{u \in U^{'}} v(u) \geq K$.
We construct a matching game $G$ such that $G$ has a nonempty stable set 
if and only if $I$ has a solution.
Let $M = \{x_1, \ldots, x_n, m_1, m_2\}$, $W = \{y_1, \ldots, y_n, w\}$, and $\Pi$ with non-zero entries:
\begin{itemize}
\item $\Pi(x_i, y_i | m_1) = - s(u_i)$ and $\Pi(x_i, y_i | m_2) = v(u_i)$, $\forall i \in \{1, \ldots, n\}$
\item $\Pi(m_1, w | m_1) = - B$
\item $\Pi(m_2, w | m_2) = K - \sum_{u_i \in U} v(u_i)$
\item $\Pi(x_j, w | x_j) = -1$ and $\Pi(m_i, y_j | y_j) = -1$, $\forall i \in \{1,2\}$ and $\forall j \in \{1, \ldots, n\}$
\end{itemize}
Since the instances with $\sum_{u_i \in U} v(u_i) < K$ are trivially solvable, we are interested in those cases where
$\sum_{u_i \in U} v(u_i) \geq K$, and so $\Pi(m_2, w | m_2) \leq 0$.
If $I$ has a solution $U^{'}$, we claim the following matching belongs to the stable set:
$\A = \left(\bigcup_{u_i \in U^{'}} \{(x_i,y_i)\}\right) \cup \{(m_1), (m_2), (w)\} \;(1)$.
First note that $U^{'}$ satisfies the knapsack conditions, and so the utilities of the agents in $\A$ are:
$u(m_1, \A) = \sum_{u_i \in U^{'}} \Pi(x_i, y_i | m_1) = \sum_{u_i \in U^{'}} -s(u_i) \geq -B$, 
$u(m_2, \A) = \sum_{u_i \in U^{'}} \Pi(x_i, y_i | m_2) = \sum_{u_i \in U^{'}} v(u_i) \geq K$,
and 
$u(x_i, \A) = u(y_i, \A) = u(w, \A) = 0, \forall i \in \{1, \ldots, n\}$.

All the agents, except possibly for $m_1$ and $m_2$, obtain their best possible utility in $\A$. Thus any blocking coalition, 
$B$, would have to 
contain at least one of $m_1$ and $m_2$. Since blocking requires each member of $B$ to perform an action, 
it follows that $m_1$, $m_2$ can only be involved in blocking $\A$ by forming a new match. Recall
that $\Pi(m_i, y_j|y_j) = -1$, and so agents $y_j$ will never accept a match with either $m_1$ or $m_2$.
Thus the only matches that $m_1$ and $m_2$ could form, if deviating, are $(m_1, w)$ and $(m_2, w)$, respectively. 

The best utility that $m_1$ can get when matched with $w$ is attained in $\A_{1}$ $=$ $\{(m_1, w)$, $(x_1)$, $\ldots$, $(x_n)$, 
$(y_1)$, $\ldots$, $(y_n)$, $(m_2)\}$:
$u(m_1, \A_{1}) = \Pi(m_1, w | m_1) = -B \leq u(m_1, \A) \; (2)$, 
where the candidate for the blocking coalition is
$B_1 = \{m_1, w\}$ $\cup$ $\left(\bigcup_{u_i \in U^{'}} \{x_i,y_i\}\right)$.

The best utility that $m_2$ can get when matched with $w$ is attained in 
$\A_{2}$ $=$ $\{(m_2, w)$, $(x_1, y_1)$, $\ldots$,
$(x_n, y_n)$, $(m_1)\}$:
$u(m_2, \A_{2}) = \Pi(m_2, w | m_2) + \sum_{i=1}^{n} \Pi(x_i, y_i | m_2)
= \left(K - \right.$ \\
$\left.\sum_{u_i \in U} v(u_i) \right) + \sum_{u_i \in U} v(u_i) = K \leq u(m_2, \A) \; (3)$,
where the candidate for the blocking coalition is 
$B_2$ $=$ $\{m_2, w\}$ $\cup$ $\left(\bigcup_{u_i \not \in U^{'}} \{x_i,y_i\}\right)$.
From inequalities $(2)$ and $(3)$, $m_1$ and $m_2$ cannot improve by deviating from $\A$. Since the other agents have no 
incentive to deviate, it follows that $\A$ belongs to the stable set.

Conversely, if the stable set of $G$ is non-empty, let $\A$ be a stable matching. First note that $\A$ must satisfy
neutral individual rationality, and so it cannot contain any match with negative value for one of the endpoints. 
Thus the only non-zero matches
that can be included in $\A$ are a subset of $\left(\bigcup_{u_i \in U} \{(x_i,y_i)\}\right)$. 
In addition, any matches of the form $(x_i, y_j)$, with $i \neq j$, can be removed from $\A$ without losing stability.
Thus without loss of generality, $\A$ can be written as in Equation $(1)$ 
for some $U^{'} \subseteq U$.
From $\A$ stable, coalitions $\{(m_1, w)\}$ and
$\{(m_2, w)\}$ are not blocking, thus Inequalities $(2)$ and $(3)$ hold.
Equivalently, the knapsack conditions are satisfied,  and so $U^{'}$ is a solution.
\end{proof}


\begin{theorem}\label{thm:membership_neutral_many}
Checking neutral stable set membership is coNP-complete.
\end{theorem}
\begin{proof}
We show that the complementary problem, of deciding whether a matching does not belong 
to the stable set of a game, is $NP$-complete. Given matching $\A$, one can nondeterministically guess
pair $\langle B, \A^{'} \rangle$ such that matching $\A$ is blocked by coalition $B$ through matching $\A^{'}$. 
To verify that $\langle B, A^{'} \rangle$ is blocking for $\A$, it is sufficient to compute, for each $z \in B$, its utility in $\A$
and $\A^{'}$ (assuming that $N \setminus B$ does not react to the deviation).

To prove hardness, we provide a reduction from \textit{Knapsack}~\cite{GareyJohnson}. A Knapsack instance has the form
$I = \langle U, s, v, B, K \rangle$, where $U = \{u_1, \ldots, u_n\}$ is a finite set, $s(u) \in \mathbb{Z}^{+}$
the size of element $u \in U$, $v(u) \in \mathbb{Z}^{+}$ the value of element $u \in U$, $B \in \mathbb{Z}^{+}$ a size constraint,
and $K \in \mathbb{Z}^{+}$ a value goal.
Let $G = (M, W, \Pi)$ be a matching game with $M = \{x_1, \ldots, x_n, m_1, m_2\}$, 
$W = \{y_1, \ldots, y_n, w_1, w_2\}$, and $\Pi$ with non-zero entries:
\begin{itemize}
\item $\Pi(x_i, y_i | m_1) = -s(u_i)$, $\Pi(x_i, y_i | w_1) = v(u_i)$, $\forall i \leq n$ 
\item $\Pi(m_2, w_2 | m_1) = -B - \varepsilon$ and $\Pi(m_2, w_2 | w_1) = K - \varepsilon$, for some $0 < \varepsilon < 1$
\item $\Pi(m_1, w |w) = \Pi(m, w_1 | m)= -1$, for all $w \in W \setminus \{w_1\}$ and $m \in M \setminus\{m_1\}$
\end{itemize}
Let $\A = \{(m_2, w_2), (m_1), (w_1), (x_1), \ldots, (x_n), (y_1), \ldots, (y_{n})\}$,
with utilities:
\begin{itemize}
\item $u(x_i, \A) = \Pi(m_2,w_2|x_i) = 0$ and $u(y_i, \A)$ $=$ $\Pi(m_2$, $w_2 | y_i) = 0$, $\forall i \in \{1, \ldots, n\}$
\item $u(m_1, \A) = \Pi(m_2, w_2|m_1) = -B - \varepsilon$
\item $u(w_1, \A) = \Pi(m_2, w_2|w_1) = K - \varepsilon$
\item $u(m_2, \A) = \Pi(m_2,w_2|m_2) = 0$, $u(w_2, \A)$ $=$ $\Pi(m_2, w_2 | w_2) = 0$
\end{itemize}
All the agents except $m_1$ and $w_1$ obtain their maximum utility in $\A$. In addition, $m_1$ or $w_1$ can only block
by forming the match $(m_1, w_1)$, since all other matches with one of these agents as an endpoint are unfeasible.
We claim that $I$ has a solution if and only if $\A$ has a blocking coalition.
If $I$ has a solution $U^{'} \subseteq U$, consider the grand coalition $N$ and matching:
$\A^{'} = \left(\bigcup_{u_i \in U^{'}} \{(x_i, y_i)\} \right) \cup \left(\bigcup_{u_i \not \in U^{'}} \{(x_i, y_{i+1}), (x_{i+1}, y_i)\} \right)$ 
$\cup \{(m_1, w_1), (m_2), (w_2)\}$,
where $x_{n+1} = x_1$ and $y_{n+1}=y_1$.
The utilities of the agents in $\A^{'}$ are:
\begin{itemize}
\item $u(x_i, \A^{'}) = 0$ and $u(y_i, \A^{'}) = 0$, $\forall i \in \{0, \ldots, n\}$
\item $u(m_1, \A^{'}) =- \sum_{u_i \in U^{'}} s(u_i) \geq - B > - B - \varepsilon = u(m_1, \A)$
\item $u(w_1, \A^{'}) = \sum_{u_i \in U^{'}} v(u_i) \geq K > K - \varepsilon = u(w_1, \A)$
\item $u(m_2, \A^{'}) = u(w_2, \A^{'}) = 0$
\end{itemize}
Thus the grand coalition is blocking under neutrality, since it can form matching $\A^{'}$ and (weakly) improve the utilities of 
all its members by doing so.

Conversely, assume $\A$ is blocked by a coalition $B$ through a matching $\A^{'}$.
Note that all the agents except $m_1$ and $w_1$, obtain their maximum possible utility in $\A$.
Player $m_1$ cannot block without agent $w_1$, since all other matches $(m_1, w)$, for $w \neq w_1$ are unfeasible;
similarly for $w_1$.
Thus any blocking coalition must include agents $m_1, w_1$ and match $(m_1, w_1)$. In addition, for one of 
$m_1$, $w_1$ to strictly improve after the deviation, the edge $(m_2, w_2)$ must be removed from $\A^{'}$.
The conditions for improvement in $\A^{'}$ are:
$(4) \; u(m_1, \A^{'}) = \sum_{(x_i, y_i) \in \A^{'}} \Pi(x_i, y_i|m_1) = \sum_{(x_i, y_i) \in \A^{'}} - s(u_i) \geq u(m_1, \A) = - B - \varepsilon$
and
$(5)\; u(w_1, \A^{'}) = \sum_{(x_i, y_i) \in \A^{'}} \Pi(x_i, y_i | w_1) = \sum_{(x_i, y_i) \in \A^{'}} v(u_i) \geq u(w_1, \A) = K - \varepsilon$.
Since $v(u_i)$ and $s(u_i)$ are integers, Inequalities $(4)$ and $(5)$
are equivalent to 
$\sum_{(x_i, y_i) \in \A^{'}} s(u_i) \leq B$ and $\sum_{(x_i, y_i) \in \A^{'}}$ $v(u_i)$ $\geq$ $K$. Thus 
$U^{'} = \{u_i \in U \, | \, (x_i, y_i) \in \A^{'}\}$ is a solution for the $I$ instance.
\end{proof}

\vspace{-2mm}
\subsection{Optimistic Stability}
For optimistic stability we prove the following theorems.


\begin{theorem}\label{thm:nonempty_opt_many}
Checking nonemptiness of the optimistic stable set is NP-complete.
\end{theorem}
\begin{proof}
First note that checking nonemptiness of the optimistic setwise-stable set is in NP. Given a matching $\A$,
the utility of each agent $z \in N$ in $\A$ can be computed in
$O(n^2)$. In addition, for every pair $(m,w) \in (M,W)$, we can again compute in $O(n^2)$ the best case 
utilities of $m$ and $w$ when matched with each other.
Verifying if $\A$ is stable can be done by iterating over all pairs $(m,w)$ and checking in $O(1)$ if both $m$
and $w$ can (weakly) improve by deviating under optimism, compared to their current utility in $\A$.
The reduction is similar to that in Theorem 1. 
\end{proof}


A problem is \textit{fixed parameter tractable (FPT)} with respect to some parameter if there exists an algorithm
which given an instance of size $n$ and parameter value $k$, computes a solution in $O(f(k)n^{O(1)})$, where $f$
is a computable function. Non-emptiness of the optimistic stable set can be determined in polynomial time by imposing
some constraints on the matches. We say that a match $(m,w)$ is \emph{feasible controversial} if $\Pi(m,w|m) = \Pi(m,w|w)=0$,
and there exist distinct agents $i_1$ and $i_2$ such that $\Pi(m,w|i_1) > 0$ and $\Pi(m,w|i_2) < 0$.

\begin{theorem}\label{thm:opt_fpt}
Checking non-emptiness of the optimistic stable set is FPT with respect to the number of feasible controversial matches.
\end{theorem}
\begin{proof}
Any matching, $\A$, belonging to the optimistic stable set must include every match in the set:
$R = \{(m,w) : \Pi(m,w|m) \geq 0, \Pi(m,w|w) \geq 0$, and $\Pi(m,w|m) + \Pi(m,w|w) > 0\}$
and exclude every match in the set $F = \{ (m,w) : \Pi(m,w|m) < 0$ or $\Pi(m,w|w) < 0\}$. 
In addition, it is always safe to exclude matches from the set $FN = \{ (m,w) : \Pi(m,w|m) = \Pi(m,w|w) = 0$ and 
$\Pi(m,w|z) \leq 0, \forall z \neq m,w$ and the inequality is strict for some $z\}$, and to include the set
$FP = \{ (m,w) : \Pi(m,w|m) =\Pi(m, w|w) = 0$ and $\Pi(m,w|z) \geq 0, \forall z \neq m,w\}$.
Thus if the optimistic stable set is non-empty, it contains a matching with all the matches in $R$ and $FP$, but none 
of $FN$ or $F$. Parameterizing the remaining edges (i.e. feasible controversial), gives an FPT algorithm for checking
non-emptiness of the optimistic stable set.
\end{proof}

\begin{theorem}\label{thm:opt_mem_many}
Checking membership to the optimistic stable set is in P.
\end{theorem}
\begin{proof}
Let $\A$ be a matching which can be blocked by a coalition $B$ through some deviation $\A^{'}$.
Assume there exists agent $z \in B$ which can strictly improve the optimistic estimation of
utility only by cutting matches in $\A^{'}$. Then $z$ can deviate alone by cutting the same matches as in $\A^{'}$
while expecting that the rest of the agents in $B$ will perform deviation $\A^{'}$ (including initiating matches with $z$
as an endpoint). Otherwise, any deviator which strictly improves utility forms a new match. Let $z$ be such an agent 
and $(z, z^{'})$ a new edge in $\A^{'}$. Then coalition $\{z, z^{'}\}$ can block by forming the edge $(z, z^{'})$,
since both agents $z$ and $z^{'}$ expect that the edges in $\A^{'}$ will form after the deviation.
\end{proof}

We can characterize optimistic stable sets as follows.
\begin{theorem} \label{thm:opt_setwise_stable}
Any matching in the optimistic stable set is a union of two disjoint matchings 
$(M^{'}, W^{'}) \cup (M^{''}, W^{''})$, where $M^{'} \cup M^{''} = M$, $W^{'} \cup W^{''} = W$, and every 
agent in $(M^{'}, W^{'})$ obtains their highest possible utility, while $(M^{''}, W^{''})$ is the empty matching.
\end{theorem}
\begin{proof}
Let $\A$ be a matching in the optimistic setwise-stable set and $z$ an agent . Since $\A$ is stable, $z$ is not blocking. 
Thus it must be that $z$ is either unmatched, case in which he cannot deviate, or $z$ is matched but already obtains the
highest possible utility, and so has no incentive to deviate. Thus $N$ can be partitioned in two subsets, $N^{'}$, the
agents that obtain in $\A$ their maximal utility, and $N^{''}$, the unmatched agents.
\end{proof}

\vspace{-4mm}
\subsection{Pessimistic Stability}
Unlike the neutral stable set, the pessimistic stable set of Example~\ref{eg:paradoxical_matching} is nonempty
when $\delta \gg \varepsilon$.
However, the pessimistic stable set can still be empty as can be seen from the following example.

\begin{example}
Let $G=(M, W, \Pi)$ with $M = \{x_1$, $x_2$, $m_1$, $m_2\}$, 
$W = \{y_1$, $y_2$, $w_1$, $w_2\}$, and $\Pi(x_1, y_1|m_1)=-3$,
$\Pi(x_2, y_2| m_1)=-5$, $\Pi(x_1, y_1|m_2) = 2$, $\Pi(x_2, y_2 | m_2) = 10$, 
$\Pi(m_1, w_1 | m_1) = -4$, $\Pi(m_2, w_2 | m_2) = -7$,
and $\Pi(m_i, y_j | y_j) = -1$, $\Pi(x_i, w_j|x_i)=-1$,
$\forall i, j \in \{1,2\}$.
\end{example}

For pessimistic stability we have the following results with respect to the complexity of computing stable outcomes.

\begin{theorem}\label{thm:nonempty_pessimistic_many}
Checking nonemptiness of the pessimistic stable set is NP-hard.
\end{theorem}

\begin{theorem}\label{thm:membership_pessimistic_many}
Checking pessimistic stable set membership is coNP-complete.
\end{theorem}

Our results for many-to-many matchings are summarized in Table 1. 
\vspace{-2mm}
\begin{table} \label{tab:many_many1}
\caption{Set Stability in Many-to-Many Matchings}
\centering
\begin{tabular}{||@{}c@{}||@{}c@{}|@{}c@{}|@{}c@{}||}
\hline \hline
\textit{Stable Set} & \textit{Neutral} & \textit{Optimistic} & \textit{Pessimistic} \\ \hline
\textit{Membership} & \textit{coNP-complete} & \textit{P} & \textit{coNP-complete} \\ \hline
\textit{Emptiness} & \textit{NP-hard} & \textit{NP-complete} & \textit{NP-hard} \\
\hline \hline
\end{tabular}
\end{table}
\vspace{-2mm}

\vspace{-2mm}
\subsection{Relationship Between the Stability Concepts}

In this section we analyze how the stable sets are related to each other.
Given a many-to-many matching game $G$, 
denote by $\mathcal{O}\mbox{-}set(G), \; \mathcal{N}\mbox{-}set(G)$,
and $\mathcal{P}\mbox{-}set(G)$ the optimistic, neutral, and pessimistic stable sets.
\begin{theorem}
Given any matching game $G$, the following inclusions hold
$\mathcal{O}\mbox{-set}(G) \subseteq \mathcal{N}\mbox{-set}(G) \subseteq \mathcal{P}\mbox{-set}(G)$.
\end{theorem}
\begin{proof}
Let $B$ denote a potentially blocking coalition.
If a matching belongs to the optimistic set, then $B$ cannot hope to improve even when the rest of the agents organize themselves in the best possible way for $B$. Thus the optimistic set is included in all the other stable sets. If a matching belongs to the pessimistic set, then $B$ cannot improve when the rest of the society will punish them maximally for the deviation. Under the other stability concepts, $B$ assumes a potentially better reaction from $N \setminus B$. Thus if $B$'s deviation is not profitable in the optimistic or neutral scenarios, it is also not profitable in the pessimistic scenario. Hence the pessimistic set contains all the other sets.
\end{proof}

\vspace{-4mm}
\section{One-to-One Matchings}

We analyze one-to-one matchings from a neutral, optimistic and pessimistic stand point just as was done with the many-to-many matchings. However, before we proceed we comment on the choice of stability concept used for one-to-one matchings. In particular, in one-to-one matchings, the setwise stable set coincides exactly with the core since agents in a deviating group must necessarily cut all matches with non-members of the deviating group due to the constraint on the number of allowable partners.
Thus, while we continue to use the term stable set to mean the setwise stable set, readers more comfortable with the core solution concept can apply it in this section.  We also investigate pairwise stability and explicitly state we are looking at pairwise stability when warranted.

\begin{definition}[Pairwise Stable]
Given a matching game $G=(M,W,\Pi)$ and a matching $\A$, $\A$ is \emph{pairwise stable} if there does not exist a blocking coalition of size one or two.
\end{definition}
In the context of one-to-one matchings, a blocking coalition of size one is equivalent to one agent that can improve utility by cutting some (or all) of his matches. 
A blocking coalition of size two is equivalent to two agents that can form a new match with each other while possibly cutting their previous match (if any), or who can coordinate to cut their existing matches without forming a new match with each other.

\vspace{-2mm}
\subsection{Neutral Stability}
The hardness results in one-to-one matchings parallel the ones in the many-to-many setting.
\begin{theorem}\label{thm:neutral_core_member_one}
Checking neutral stable set membership is coNP-complete.
\end{theorem}

\begin{theorem}\label{thm:neutral_core_nonempty_one}
Checking nonemptiness of the neutral stable set is NP-hard.
\end{theorem}

We also study pair-wise stability under the neutral assumption. We first note that there is a separation between the neutral pairwise stable set and the neutral setwise stable set.
\begin{example}
Let $G = (M, W$, $\Pi)$, where $M = \{m_1, m_2\}$, $W = \{w_1, w_2\}$, and $\Pi(m_i, w_j|m_i)$ $=$ $\Pi(m_i, w_j | w_j)$ $=$ $-\varepsilon$,
$\Pi(m_i, w_j | z)$ $=$ $W \gg \varepsilon > 0$, where $z \in N \setminus \{m_i, w_j\}$, $\forall i, j \in \{1,2\}$.
The neutral stable set of $G$ is empty, while the ``empty'' matching
is pairwise stable.

\end{example}

While the neutral pairwise stable set may be empty, under certain conditions we can compute such a stable matching in polynomial time.
\begin{theorem} \label{thm:neut_stable}
A neutral pairwise stable matching can be computed in polynomial time when $\Pi$ is non-negative.
\end{theorem}
\begin{proof}
Let $\A$ be the matching returned by running the Gale-Shapley algorithm by ignoring externalities, 
and assume by contradiction $\A$ is not stable.
Then there exists  deviation $(m,w^{'})$, where $m, w$ are matched in $\A$ with $w^{'}, m^{'}$, respectively.
Let $ext(m, \A)$ denote the value obtained by $m$ in $\A$ from externalities, and $E_{m}(w) = \Pi(m,w|m)$.
From $\{m,w^{'}\}$ blocking,
$u(m, \A) = \Pi(m,w|m) + ext(m, \A) < \Pi(m,w^{'}|m) + ext(m, \A) - \Pi(m^{'}, w^{'}|m)$ and
$u(w^{'}, \A)=\Pi(m^{'}, w^{'} | w^{'}) + ext(w^{'}, \A) < \Pi(m,w^{'}|w^{'})+ ext(w^{'}, \A)- \Pi(m,w|w^{'})$.
Equivalently,
$(i)\;\Pi(m,w|m) < \Pi(m,w^{'}|m) - \Pi(m^{'},w^{'}|m) \leq \Pi(m,w^{'}|m)$ and
$(ii)\;\Pi(m^{'},w^{'}| w^{'}) < \Pi(m,w^{'}|w^{'})- \Pi(m,w|w^{'}) \leq \Pi(m,w^{'}|w^{'})$.
From $(i), (ii)$, $E_{m}(w) < E_{m}(w^{'})$ and $E_{w^{'}}(m^{'}) < E_{w^{'}}(m)$,
and so $(m,w^{'})$ is blocking under the preferences given by $E$.
Contradiction, since $\A$ is stable on $(M,W,E)$.
\end{proof}

However, the neutral pairwise-stable set can be empty when $\Pi$ can
have negative entries (Example~\ref{eg:empty_neutral_pairwise}).

\begin{example} \label{eg:empty_neutral_pairwise}
Let $G = (M, W, \Pi)$, where $M = \{m_1, m_2\}$, $W = \{w_1, w_2\}$, and $\Pi$ as follows:
$\Pi(m_i, w_j|m_i) = \Pi(m_i, w_j | w_j) = 1$, for all $i, j \in \{1,2\}$, and
$\Pi(m_i, w_j | z) = -1$, for all $i, j \in \{1,2\}$ and $z \neq m_i, w_j$.
All the edges have positive values for their endpoints, and so
the only candidates for neutral pairwise stability are:
$\A_1 = \{(m_1, w_1), (m_2, w_2)\}$ and $\A_2 = \{(m_1, w_2), (m_2, w_1)\}$.
However, matching $\A_1$ is blocked by the pair $(m_1, w_2)$, and $\A_2$ is blocked by $(m_1, w_1)$.
\end{example}

\vspace{-2mm}
\subsection{Optimistic Stability}


\begin{theorem}
Checking nonemptiness of the optimistic stable set is NP-complete, even when $\Pi$ is non-negative.
\end{theorem}
\begin{proof}
First note that similarly to the optimistic setwise-stable set, checking nonemptiness of the optimistic stable set is in NP. Given a matching $\A$,
the utility of each agent $z \in N$ in $\A$ can be computed in
$O(n^2)$. In addition, for every pair $(m,w) \in (M,W)$, we can again compute in $O(n^2)$ the best case 
utilities of $m$ and $w$ when matched with each other.
Verifying if $\A$ is stable can be done by iterating over all pairs $(m,w)$ and checking in $O(1)$ if both $m$
and $w$ can (weakly) improve by deviating under optimism, compared to their current utility in $\A$.
The same reduction as in Theorem 1 applies, by noting that the weights are constructed
such that any feasible matching (from the perspective of the $x_i$ and $y_i$ agents) is one-to-one.
\end{proof}

For the next theorem, we consider the weak optimistic stable set. Weak stability requires that all members of a blocking coalition strictly 
improve their utility in order for the deviation to take place.

\begin{theorem}
The optimistic pair-wise stable set is equivalent to the weak optimistic stable set.
\end{theorem}
\begin{proof}
Let $\A$ be a one-to-one matching. We show that if $\A$ has a blocking coalition $B$ under optimism, 
then there exists a blocking singleton or pair. We consider two cases.
If there exists $z \in B$ which improves strictly in $\A^{'}$ only by severing a match. Then agent $z$ can deviate
alone, by simply cutting the same match as in $\A^{'}$, and expecting (optimistically) that the rest of the agents 
will react as in $\A^{'}$. Otherwise, there exists $z \in B$ which improves strictly by forming a new match, say $(z, z^{'})$,
and possibly severing an existing match. Then coalition $\{z, z^{'}\}$ is blocking.
Thus if $\A$ does not belong to the weak optimistic stable set, it is also not pairwise stable under optimism. The reverse 
direction is clear, and so the optimistic pairwise stable set is equivalent to the weak optimistic stable set.
\end{proof}

\begin{corollary}
Checking nonemptimess of the optimistic pairwise set is NP-complete, even when $\Pi \geq 0$.
\end{corollary}

\vspace{-2mm}
\subsection{Pessimistic Stability}

For pessimistic stability, we obtain similar results to the ones for neutrality, and in addition,
we consider a relaxation of pairwise stability which is well-defined for non-negative $\Pi$.
\begin{theorem}
Checking pessimistic stable set membership is coNP-complete.
\end{theorem}

\begin{theorem}
Checking nonemptiness of the pessimistic stable set is NP-hard.
\end{theorem}

Pessimistic pairwise stable matchings can be computed in polynomial time when $\Pi \geq 0$.
\begin{theorem}
A pessimistic pairwise stable set can be computed in polynomial time when $\Pi$ is non-negative.
\end{theorem}
\begin{proof}
A neutral pairwise stable matching can be computed in polynomial time when $\Pi \geq 0$,
and any matching satisfying neutral pairwise stability also satisfies pessimistic pairwise
stability.
\end{proof}

In one-to-one matchings with non-negative $\Pi$, we also consider a restricted notion of pessimism. 
Namely, any matching containing singletons (unmatched agents) can weakly improve 
everyone's utility by pairing the singletons with each other.
Based on this observation, the pessimistic deviators can have a less extreme attitude, and assume 
that while the rest of the agents may punish them for the deviation, they will not stay unmatched 
in order to do so. In other words, a blocking coalition, $B$, assumes that the agents in $N \setminus B$
form the worst possible matching for $B$, among all the matchings of size $\min\left(|M \cap \left(N \setminus B\right)|, |W \cap \left(N \setminus B\right)|\right)$.

\begin{theorem}\label{thm:pess_stable}
A restricted pessimistic pairwise stable matching can be computed in polynomial time.
\end{theorem}
\begin{proof}
Let $\A$ be the matching returned by Algorithm 1 and assume by contradiction $\A$ is unstable under
pessimistic pairwise stability. Then there exists deviating pair $(m,w^{'})$, where $m,w$ are matched in $\A$ with
$m^{'}, w^{'}$, respectively.
Then it must be the case that for any possible matching $\A^{'}(m,w^{'})$ that includes the pair $(m,w^{'})$, both $m$ and $w^{'}$
are better off in $\A^{'}(m,w^{'})$ than in $\A$.
Equivalently, $E^{-}_{m}(w^{'}) > u(m, \A) \geq E^{-}_{m}(w)$ and $E^{-}_{w^{'}}(m) > u(w^{'}, \A) \geq E^{-}_{w^{'}}(m^{'})$. However,
$\A$ is stable under $E^{-}$, and so for any $(m,w^{'}) \not \in \A$, either $E^{-}_{m}(w^{'}) \leq E^{-}_{m}(w^{'})$ or
$E^{-}_{w^{'}}(m^{'}) \leq E^{-}_{w^{'}}(m)$, contradiction. Thus $\A$ is stable.
\end{proof}

\vspace{-4mm}
\begin{algorithm}[htbp]\label{alg1}
\caption{Pessimistic Pairwise Stability}
   \ForAll{$z \in N$} {
       \ForEach{$t \in Opposite(z)$} {
           $M^{'} \gets M \setminus\{z,t\}$\\
           $W^{'} \gets W \setminus\{z,t\}$\\
           \ForEach{$(m,w) \in M^{'} \times W^{'}$} {
               $\psi(m,w) \leftarrow \Pi(m,w|z)$\\
           }
           $\A^{-} \gets \mbox{Min-Matching}(M^{'}, W^{'}, \psi)$\\
           $E^{-}_{z}(t) \leftarrow \Pi(z,t|z)$\\
           \ForEach{$(m,w) \in \A^{-}$} {
               $E^{-}_{z}(t) \leftarrow E^{-}_{z}(t) + \psi(m,w)$\\
           }
       }
   }
   \textbf{return} $\mbox{Gale-Shapley}(M, W, E^{-})$\\
\end{algorithm}

Our results for one-to-one matchings are summarized in Table 2 and Table 3.
\begin{scriptsize}
\begin{table}
 \hspace{0.01\textwidth}
 \begin{minipage}{0.58\textwidth}
  \centering
\caption{\textnormal{Set Stability in 1-1 Matchings}}
  \label{tab:caption1}
\begin{tabular}{||@{}c@{}||@{}c@{}|@{}c@{}|@{}c@{}||}
\hline \hline
\textit{Stable Set} & \textit{Neutral} & \textit{Optimistic} & \textit{Pessimistic} \\ \hline
\textit{Membership} & \textit{coNP-complete} & \textit{P} & \textit{coNP-complete} \\ \hline
\textit{Emptiness} & \textit{NP-hard} & \textit{NP-complete} & \textit{NP-hard} \\
\hline \hline
\end{tabular}
 \end{minipage}
 \hspace{0.01\textwidth}
 \begin{minipage}{0.58\textwidth}
  \centering
\caption{\textnormal{Pairwise Stability in 1-1 Matchings ($\Pi \geq 0$)}}
  \label{tab:caption2}
\begin{tabular}{||@{}c@{}||@{}c@{}|@{}c@{}|@{}c@{}|@{}c@{}||}
\hline \hline
\textit{Pairwise Stability} & \textit{Neutral} & \textit{Optimistic} & \textit{Pessimistic} \\ \hline
\textit{Membership} & \textit{P} & \textit{P} & \textit{P} \\ \hline
\textit{Emptiness} & \textit{P} & \textit{NP-complete} & \textit{P}\\ \hline
\hline \hline
\end{tabular}
 \end{minipage}
\end{table}
\end{scriptsize}

\vspace{-4mm}
\section{Related Work}
Klaus \emph{et al}~\cite{KlausWalzl09:StableManyToManyContracts} consider pairwise and setwise stability, with the weak and strong variants, in many-to-many matching markets.
Echenique \emph{et al}~\cite{Echenique06:TheoryStability} study several solution concepts,
such as the setwise-stable set, the core, and the bargaining set in many-to-many matchings.
These models do not consider externalities or the boundedness of the agents as they use exponential
preference profiles.
Externalities in the classical marriage problem have been introduced by Sasaki and Toda~\cite{SasakiToda96:MarriageExternalities}, and in one-to-many models
by Dutta and Masso~\cite{DuttaMasso97:MatchingsPreferencesColleagues}, both for complete preference profiles. 
Hafalir \cite{hafalir08} studied externalities in marriage problems,
Weighted preferences have been introduced in matchings by Pini \emph{et al.}~\cite{WalshWeighted}, in which the agents rank each other using a numerical value. However, they study solution concepts different from ours, such as $\alpha$-stability and link-additive stability, and do not consider externalities.

\vspace{-4mm}
\section{Discussion}

In this work we introduced a compact model for matching problems with externalities. 
In many-to-many matchings, we considered neutral, optimistic, and pessimistic reasoning under the setwise-stable set solution concept.
In one-to-one matchings, we considered neutral, optimistic, and pessimistic reasoning under 
the setwise-stable set and pairwise stability solution concepts.
In both the many-to-many and one-to-one settings, we studied the computational complexity of finding stable outcomes and provided
both hardness results and polynomial algorithms.

This work can be extended in several ways. A coherent theory of stability for matchings with externalities remains to be further developed. 
The existing work on the stability of matchings assumes complete preference profiles, which require exponential space and reasoning 
abilities on behalf of the agents. 
Such modelling is not realistic when dealing with bounded rational agents.
Furthermore, externalities are important in many other domains, including network formation games. 
Effects due to externalities are known to appear in networks~\cite{Jackson08:SocialEconomicNetworks}, however we are not aware of 
a theory of stability for network formation games that explicitly accounts for externalities.
Since such games involve very large number of agents, a compact model of externalities is crucial.
\section{Acknowledgements}
Simina Br\^{a}nzei was supported in part by the Sino-Danish Center for the Theory of Interactive Computation, funded by the Danish National
Research Foundation and the National Science Foundation of China (under the grant 61061130540) and by the Center for research in the
Foundations of Electronic Markets (CFEM), supported by the Danish Strategic Research Council.



\section{Appendix}

\setcounter{theorem}{6}
\begin{theorem}\label{thm:nonempty_pessimistic_many}
Checking nonemptiness of the pessimistic stable set is NP-hard.
\end{theorem}
\begin{proof}
The same reduction as in Theorem 1
applies, by noting that a coalition is blocking 
for the constructed matching under pessimistic reasoning if and only if it is blocking under neutral reasoning.
\end{proof}

\begin{theorem}\label{thm:membership_pessimistic_many}
Checking pessimistic setwise-stable set membership is coNP-complete.
\end{theorem}
\begin{proof}
The reduction from Theorem 2 applies, by noting that in matching $\A$
of Theorem 2, a coalition is blocking under neutral reasoning if and only if
it is blocking under pessimistic reasoning.
\end{proof}

\setcounter{theorem}{9}
\begin{theorem}\label{thm:neutral_core_member_one}
Checking neutral stable set membership is coNP-complete.
\end{theorem}
\begin{proof}
The same reduction as in Theorem 2 applies, by noting that the 
weights are set such that any feasible matching (from the perspective of the $x_i$ and $y_i$ 
agents) is one-to-one. In other words, a coalition can block the matching constructed in 
Theorem 2 if and only if it is a one-to-one matching.
\end{proof}

\begin{theorem}\label{thm:neutral_core_nonempty_one}
Checking nonemptiness of the neutral stable set is NP-hard.
\end{theorem}
\begin{proof}
The same reduction as in Theorem 1 applies, by noting that the
weights are set up such that any feasible matching (from the perspective of the $x_i$ and $y_i$
agents) is one-to-one, and so a matching of the game in Theorem 1 is stable if and only if
it is a one-to-one matching.
\end{proof}

\setcounter{theorem}{14}
\begin{theorem}
Checking pessimistic stable set membership is coNP-complete.
\end{theorem}
\begin{proof}
The reduction is similar to that of Theorem 1, except for 
the addition of several dummy agents to ensure that the grand coalition can always block 
in the one-to-one setting when the \textit{Knapsack} instance has a solution.
Given $I = \langle U, s, v, B, K \rangle$,
let $G = (M, W, \Pi)$ such that
$M = \{x_1, \ldots, x_{2n}, m_1, m_2\}$,
$W = \{y_1, \ldots, y_{2n}, w\}$,
and $\Pi$ with non-zero entries:
\begin{itemize}
\item $\Pi(x_i, y_i | m_1) = \Pi(x_i, y_{n+i}|m_1) = -s(u_i)$ and $\Pi(x_i, y_i | w) = \Pi(x_i, y_{n+i} | w) = v(u_i)$,
$\forall i \in \{1, \ldots, n\}$
\item $\Pi(m_1, w | m_1) = - B$
\item $\Pi(m_2, w | m_2) = K - \sum_{u_i \in U} v(u_i)$
\item $\Pi(x_j, w | x_j) = -1$ and $\Pi(m_i, y_j | y_j) = -1$, $\forall i \in \{1,2\}$ and $\forall j \in \{1, \ldots, 2n\}$
\end{itemize}
Let matching $\A = \{(m_1), (m_2), (w), (x_1), \ldots, (x_{2n})$, $(y_1)$,$ \ldots, (y_{2n})\}$.
Similarly to Theorem 1,
$A$ belongs to the pessimistic stable set if and only if the \textit{Knapsack} instance has a solution.
\end{proof}

\begin{theorem}
Checking nonemptiness of the pessimistic stable set is NP-hard.
\end{theorem}
\begin{proof}
The same reduction as in Theorem 1 applies, by noting that the
weights are set up such that any feasible matching (from the perspective of the $x_i$ and $y_i$
agents) is one-to-one, and so a matching of the game in Theorem 1
is stable if and only if it is a one-to-one matching.
\end{proof}

\end{document}